\documentclass[11pt]{article}

\usepackage{amsmath}
\usepackage{amsthm}
\usepackage{url}
\usepackage[bookmarks]{hyperref}
\usepackage{xspace}
\usepackage{algorithm}
\usepackage{algpseudocode}
\usepackage[margin=1.5in]{geometry}
\usepackage{paralist}

	\usepackage{color}
	\usepackage{soul}
	\definecolor{lightblue}{rgb}{.60,.60,1}
	
\usepackage{framed}
\usepackage{tikz}
\usetikzlibrary{arrows}
\usepackage{float}
  \restylefloat{figure}

\newcommand{\latticecut}{\textsc{Poset Cut}\xspace}
\newcommand{\linearcut}{\textsc{Linear Cut}\xspace}
\newcommand{\dmultiwaycut}{\textsc{Directed Multiway Cut}\xspace}
\newcommand{\umultiwaycut}{\textsc{Multiway Cut}\xspace}
\newcommand{\skewmulticut}{\textsc{Skew Separator}\xspace}
\newcommand{\dmulticut}{\textsc{Directed Multicut}\xspace}
\newcommand{\umulticut}{\textsc{Multicut}\xspace}
\newcommand{\clique}{\textsc{Clique}\xspace}
\newcommand{\dfvs}{\textsc{Directed Feedback Vertex Set}\xspace}

\newcommand{\W}{\mathsf{W}\xspace}
\newcommand{\NP}{\mathsf{NP}\xspace}
\newcommand{\PT}{\mathsf{P}\xspace}
\newcommand{\FPT}{\mathrm{FPT}\xspace}

\newcommand{\NO}{\textsc{NO}\xspace}

\newcommand{\T}{{\mathcal T}}

\newcommand{\pair}[1]{{\left\langle#1\right\rangle}}
\newcommand{\size}[1]{{\left\lvert#1\right\rvert}}

\DeclareMathOperator{\poly}{poly}

\newcommand{\sm}{\setminus}

\theoremstyle{plain}        \newtheorem{thm}{Theorem}
\theoremstyle{plain}        
\theoremstyle{plain}        \newtheorem{prop}[thm]{Proposition}
\theoremstyle{plain}        
\theoremstyle{definition}  
\theoremstyle{plain}        
\theoremstyle{plain}        \newtheorem{cor}[thm]{Corollary}
\theoremstyle{remark}    

\author{
Robert F.\ Erbacher\\
\emph{Army Research Lab}\\
\url{robert.f.erbacher.civ@mail.mil}
\and
Trent Jaeger\\
\emph{Penn State University} \\
\url{tjaeger@cse.psu.edu}
\and
Nirupama Talele\\
\emph{Penn State University}\\
\url{nrt123@psu.edu}
\and
Jason Teutsch \\
\emph{Penn State University} \\
\url{teutsch@cse.psu.edu}
}

\title{Directed Multicut with linearly ordered terminals}

\begin{document}

\maketitle

\begin{abstract}
Motivated by an application in network security, we investigate the following ``linear'' case of \dmulticut.  Let $G$ be a directed graph which includes some distinguished vertices $t_1, \ldots, t_k$.  What is the size of the smallest edge cut which eliminates all paths from $t_i$ to $t_j$ for all $i < j$?  We show that this problem is fixed-parameter tractable when parametrized in the cutset size~$p$ via an algorithm running in $O(4^p p n^4)$ time.
\end{abstract}

\section{Multicut requests as partially ordered sets}

The problem of finding a smallest edge cut separating vertices in a graph has received much attention over the past 50 years.  \dmulticut, one of the more general forms of this problem, encompasses numerous applications in algorithmic graph theory.
\begin{framed}
\begin{description}
\item [Name:] \dmulticut.
\item[Instance:] A directed graph $G$ and pairs of terminal vertices $\{(s_1, t_1), \dotsc, (s_k, t_k)\}$ from~$G$.
\item[Problem:] Find a smallest set of edges in $G$ whose deletion eliminates all paths $s_i \to t_i$.
\end{description}
\end{framed}

Special cases of the \dmulticut problem have been met with success, although the general problem has no polynomial-time solution unless $\PT = \NP$.  The classical and efficient Ford-Fulkerson algorithm \cite{FF56} solves \dmulticut for the case of a single pair of terminal vertices, yet deciding whether there exists a minimum edge cut of a given size separating both $s$ from $t$ and $t$ from $s$ in a directed graph is $\NP$-complete \cite{GVY94} as is deciding the size of a minimum edge cut separating three vertices in an undirected graph \cite{DJPSY92}.

While \dmulticut appears intractable from the perspective of $\NP$-completeness, it remains an open problem to determine whether we can find an efficient parametrized solution for \dmulticut.  In practice we can optimize our solution based on other parameters besides the input length.  In the case of \dmulticut, the relevant parameters are the number of (sets of) terminal vertices $k$ and and the size of the smallest solution, or \emph{cutset}, $p$. Formally a problem is \emph{fixed parameter tractable (FPT) in parameters $k$ and~$p$} if there exists an algorithm which, on input $x$, either gives a solution consistent with parameters~$k$ and~$p$ or correctly decides that no such solution exists in at most $f(k,p) \cdot \poly(\size{x})$ steps for some computable bound~$f$.

Some subcases of \dmulticut already have FPT solutions within the realm of fixed-parameter tractability.  Recently Kratsch, Pilipczuk, Pilipczuk, and  Wahlstr\"{o}m \cite{KPPW12} showed that \dmulticut restricted to acyclic graphs is fixed-parameter tractable when parameterized in both the size of the cutset and the number of terminals.  Chitnis, Hajiaghayi, and Marx \cite{CHM12}, on the other hand, investigated \dmulticut with restrictions of the terminal pairs.  They showed that \dmultiwaycut,  the special case of \dmulticut where all pairs of terminal vertices must be separated in both directions, is FPT when parametrized in just the size of the cutset.  In the negative direction, Marx and Razgon \cite{MR11} showed that \dmulticut is $\W[1]$-hard when parameterized the size of the cutset.  Thus an FPT solution for \dmulticut, if such an algorithm exists, most likely requires parameterization in the number of terminals in addition to the size of the cutset.  We remark that in this same paper \cite{MR11} Marx and Razgon also showed that the undirected \umulticut problem is FPT when parametrized in  the size of the cutset.  Bousquet, Daligault, and Thomass{\'e} independently achieved this same result \cite{BDT11}.

We now formalize the \latticecut problem, a subject which derives from a network security framework \cite{Pik09}.  We shall show that \latticecut is equivalent to \dmulticut with respect to fixed parameter tractability.
\begin{framed}
\begin{description}
\item [Name:] \latticecut

\item[Instance:] A directed graph $G=(V,E)$ with terminal vertices $T \subseteq V$, a partially ordered set $P$, and a surjective map $\ell : T \to P$.

\item[Problem:] Find a minimum set of edges $S \subseteq E$ so that for all terminal vertices $x,y \in T$, if there is a path from $x$ to $y$ in $(V, E \setminus S)$ then $\ell(x) \geq_P \ell(y)$.
\end{description}
\vspace{-1ex}
\end{framed}

The \latticecut problem is immediately a special case of \dmulticut.  Indeed, given an instance of \latticecut, we can read off from the poset $P$ and mapping $\ell: T \to P$ those pairs of terminals which must be separated in the \latticecut solution.  These pairs together with the original input graph give us an instance of \dmulticut such that an edge cut is a solution to the \latticecut instance if and only if it is a solution to the \dmulticut instance.  Thus if \dmulticut is fixed-parameter tractable, then so is \latticecut.  We now show that the reverse is also true.

\begin{thm} \label{thm: hop}
If \latticecut is $\FPT$, then so is \dmulticut.  In particular, given an instance of \dmulticut with $k$ terminal pairs and a permitted maximum of $p$ cuts, we can efficiently find an instance of \latticecut with at most $2k$ terminal nodes and a permitted maximum of $p$ cuts such that the \latticecut instance has a solution iff the \dmulticut instance does.
\end{thm}

\begin{proof}
Consider an instance of \dmulticut consisting of a graph $G$, forbidden terminal pairs $s_1 \not\to t_1, \dotsc, s_k \not\to t_k$, and a cutsize parameter $p$.  We define the corresponding \latticecut instance as follows.  The graph $G'$ will consist of all the nodes and edges in $G$ plus some extra nodes and edges.  For each terminal node $s_i$, add a node $a_i$ and enough paths from $a_i$ to $s_i$ so that $a_i$ and $s_i$ remain connected in any solution for the \latticecut instance.  In more detail
\begin{itemize}
\item add $p+1$ nodes $c_{i,1}, \dotsc, c_{i,p+1}$, 
\item add an edge from $a_i$ to each $c_{i,j}$, and
\item add a further edge from each $c_{i,j}$ to $s_i$.
\end{itemize}
Similarly for each terminal node $t_i$, we add a node $b_i$ and connect $t_i$ to $b_i$ with many paths: make $p+1$ new nodes $d_{i,1}, \dotsc d_{i,p+1}$, add an edge from $t_i$ to each $d_{i,j}$, and add an edge from each $d_{i,j}$ to $b_i$.  We define the poset for this \latticecut instance as follows: set $a_i$ to be greater than $b_j$ for all $i \neq j$, and all other pairs of terminal nodes are designated as incomparable.

By construction, there is a path $a_i \to b_i$ iff there is a path $s_i \to t_i$, and this condition holds even when up to $p$ edges are deleted from $G'$.  If there is a \latticecut solution on $G'$ under the given poset with at most $p$ cuts, there is a further solution which is identical but avoids cutting any paths between $a_i$ and $s_i$ or $t_i$ and $b_i$.  Hence we may assume that the solution has all its cuts inside the embedding of $G$ within $G'$.  Transferring these cuts back to the original graph $G$ gives a solution for the \dmulticut instance.  On the other hand, any solution for \dmulticut in $G$ will also be a solution for \latticecut in $G'$ because the only paths between pairs of terminal vertices in the \latticecut instance start at some $a_i$ and end at some $b_j$.
\end{proof}

Edwards, Jaeger, Muthukmaran, Rueda, Talele, Teutsch, Vijayakumar \cite{MRTVJTE12} and Jaeger, Teutsch, Talele, Erbacher \cite{EJTT13a} distilled the placement of host security mediators on a distributed system to a solution for the \latticecut problem.   They interpreted the components of a distributed system as nodes in a directed graph with edges indicating which components can communicate directly with others.  Some information traveling through a network will have high integrity, and other information will have lower integrity, and security is achieved by blocking all flows from lower integrity to higher integrity nodes.  Terminal nodes represent both the possible attack surfaces and higher integrity entities in the system, and each terminal corresponds to a specific integrity level as measured by the poset.  In this context, we can interpret \latticecut as a search for minimum intervention which mediates between all illegal information flows.

For the remainder of this paper, we will focus on the subcase of \latticecut where the poset is a chain.
\begin{framed}
\begin{description}
\item[Name:] \linearcut

\item[Instance:] A directed graph $(V,E)$ and a tuple of \emph{terminal} sets $\pair{T_1, \dotsc, T_k}$ which are subsets of $V$.

\item[Problem:] Find a smallest set of edges $S \subseteq E$ such that for any $s \in T_i$ and $t \in T_j$, if there is a path from $s$ to $t$ in $(V, E \setminus S)$, then $i \geq j$.
\end{description}
\end{framed}
That is, \linearcut wants to find a smallest edge cut which prevents every terminal set $T_i$ from flowing to $T_j$ whenever $j > i$.  We shall show that \linearcut, which is $\NP$-hard in the sense of Proposition~\ref{prop: lcnp}, is FPT when parameterized in the size of the cutset.  Rephrased in terms of posets,  Chitnis, Hajiaghayi, and Marx's algorithm \cite{CHM12} for \dmultiwaycut shows that \latticecut is FPT parametrized in the cutset size when the underlying poset is an antichain.

\section{A parameterized algorithm for Linear Cut}

We shall show that \linearcut is FPT when parametrized in the size of cutset.  Before presenting our parametrized algorithm, we first analyze the following example which illustrates why the na\"{i}ve greedy cut does not yield an optimal solution.  The graph given in Figure~\ref{fig: greedy} has three terminal vertices $t_0$, $t_1$, and $t_2$, and we would like to find a small set of edges whose removal eliminates all paths from $t_0$ to either $t_1$ or $t_2$ as well as all paths from $t_1$ to $t_2$.  Consider the greedy algorithm which uses the Ford-Fulkerson algorithm to first eliminate all paths from $t_0$ to the other terminal vertices and then again to extinguish the paths from $t_1$ to $t_2$.  A minimal edge cut from $t_0$ to the set $\{t_1, t_2\}$ has size~3, so let us assume that the algorithm chooses edges $\{a,b,c\}$.  Now a minimal edge cut from $t_1$ to $t_2$ has size~2, for example $\{h,i\}$.  Thus this greedy algorithm solves the \linearcut instance with a cut of size~5.  On the other hand, $\{d,e,f,g\}$ is a solution of size~4.

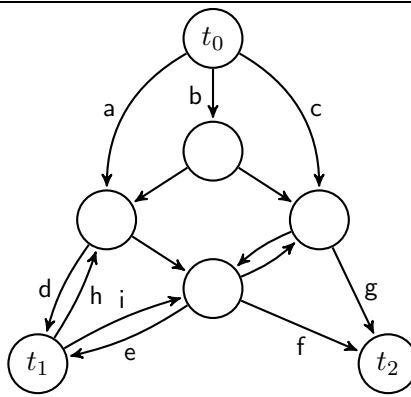
\begin{figure}[h]
\caption{The greedy algorithm is not optimal.}
\begin{center}
\begin{tikzpicture}[->,>=stealth',shorten >=1pt,node distance=2cm,
  thick,terminal/.style={circle,fill=white!20,draw}, regular/.style ={circle,fill=white!20,draw},scale = 1]

  \node[terminal] (t_0){$t_0$};
  \node[regular] (t) [below of =t_0, yshift = 0.5cm] {\phantom{$t_0$}};
  \node[regular] (l) [below left of=t, yshift = 0.5cm] {\phantom{$t_0$}};
  \node[regular] (r) [below right of=t, yshift = 0.5cm] {\phantom{$t_0$}};
  \node[regular] (b) [below right of=l, yshift = 0.5cm] {\phantom{$t_0$}};
  \node[terminal] (t_1)[below left of=l, xshift = 0.5cm, yshift = -0.5cm]{$t_1$};
  \node[terminal] (t_2) [below right of=r, xshift = -0.5cm, yshift = -0.5cm] {$t_2$};

  \path[every node/.style={font=\sffamily\small}]
    (t_0) edge [bend right] node [left] {a} (l)
          edge node [left] {b} (t)
          edge [bend left] node [right] {c} (r)
    (t) edge (l)
      edge (r)
    (l) edge (b)
    (r) edge [bend right=10] (b)
    (b) edge [bend right=10] (r)
    (l) edge [bend right=10] node [left] {d} (t_1)
    (t_1) edge [bend right=7] node [right] {h} (l)
    (t_1) edge [bend left=7] node [above] {i} (b)
    (b) edge [bend left=10] node [below] {e} (t_1)
    (r) edge node [right] {g} (t_2)
    (b) edge node [below] {f} (t_2);
\end{tikzpicture}
\end{center}
\label{fig: greedy}
\end{figure}

We now describe our parametrized solution for \linearcut.  Our algorithm either outputs a solution cut of size less $\leq p$ or returns \NO if no such cut exists.  Our construction exploits a technique used in Chen, Liu, and Lu's fixed-parameter solution \cite{CLL09} to the \umultiwaycut problem in undirected graphs which improved a result of Marx \cite{Mar06}.   A similar idea appeared earlier in Chen, Liu, Lu, O'Sullivan, and Razgon's algorithm \cite{CLLOR08} for \skewmulticut, a key step in their parametrized solution for \dfvs.  We remark that the pushing of important separators technique along the lines of \cite[Theorem~3.7]{Mar06} gives a parameterized solution for \linearcut in time $O(4^{p^3} n^{O(1)})$, and using a reduction to the \skewmulticut algorithm in \cite{CLLOR08} one can also show that \linearcut has a solution which runs in the same time as the algorithm given below, namely $O(4^p p n^4)$.

 An ($X$,$Y$)-\emph{separator} is a set of edges such that any path from $X$ to $Y$ passes through one of its members.  Our solution, Algorithm~\ref{alg: lc} proceeds in two phases.  First we handle the trivial cases where $\T = \pair{}$, $p=0$, or $T_1$ is either already separated from the other terminals or can't be separated with $p$ edge cuts (lines 1--12).  The second phase picks an edge pointing out of the $T_1$ region and checks whether making it undeleteable hurts the min size of a $(T_1, T_2 \cup \dotsb \cup T_k)$-separator.  If not we add the edge to the list of undeleteable edges, and if so we branch on the only two possibilities: either the edge belongs in the \linearcut solution or it doesn't.

The following theorem gives the main justification for this algorithm.   A set of edges is a \emph{linear cut} with respect to the $k$-tuple of terminals $\pair{T_1, \dotsc, T_k}$ if there is no path from $T_i$ to $T_j$ whenever $i<j$ once these edges have been removed.

\newpage
\ \vspace{.3in}
\begin{algorithm}[h!]
\algnewcommand\algorithmicinput{\textbf{Input:}}
\algnewcommand\Input{\item[\algorithmicinput]}
\algnewcommand\algorithmicoutput{\textbf{Output:}}
\algnewcommand\Output{\item[\algorithmicoutput]}
\algnewcommand\llet{\textbf{let}\xspace}
\algnewcommand\LC{\mathsf{LC}\xspace}
\caption{FPT algorithm for \linearcut parameterized in cutset size.}
\label{alg: lc}
\begin{algorithmic}[1]  
\Input A graph $G = (V, E)$, a $k$-tuple of terminal sets $\pair{T_1, T_2, \dotsc T_k}$ which are subsets of $V$, some undeletable edges $F \subseteq E$, and a parameter $p$.
\Statex
\Output A set of $\leq p$ edges in $E \sm F$ such that when these edges are deleted from $G$ there is no path from $T_i$ to $T_j$ for any $i<j$, if such a set of edges exists, otherwise return \NO.
\Statex
\Function{$\LC$}{$(V,E), \pair{T_1, \dotsc, T_k}, F, p$}
\State For ease of reading, let $\T = \pair{T_1, \dotsc, T_k}$.
\If {$\T = \pair{}$} \Return $\emptyset$;
\ElsIf {$p \leq 0$}
\If {for all $i <j$, $T_j$ is not reachable from $T_i$ in $G$} \Return $\emptyset$;
\Else\ \Return \NO;
\EndIf
\EndIf
\State
\llet $m$ be the size of a minimum $(T_1, T_2 \cup \dotsb \cup T_k)$-separator which does not include edges from $F$.
\If {$m > p$ or no separator exists due to undeleteable edges} \Return \NO;
\ElsIf {$m = 0$} \Return $\LC((V,E), \pair{T_2, \dotsc, T_k}, F, p)$;
\Else
\State \llet $e \in E \sm F$ be an edge with a tail reachable from $T_1$ via undeleteable edges.
\If {the size of a minimum $(T_1, T_2 \cup \dotsb \cup T_k)$-separator which does not include edges from $F \cup \{e\}$ exists and is equal to $m$,}
\State \Return $\LC ((V,E), \T , F \cup \{e\}, p)$;
\ElsIf {$\{e\} \cup \LC ((V,E \sm \{e\}), \T, F, p-1)$ or $\LC ((V,E), \T, F \cup \{e\}, p)$ is not \NO,}
\State \Return the first of these two found to have a solution;
\Else
\State \Return NO;
\EndIf
\EndIf
\EndFunction
\end{algorithmic}
\end{algorithm}

\newpage
\begin{thm} \label{thm: bobo}
Let $\pair{(V,E), \T, F, p}$ be an input to Algorithm~1, where $\T$ is an abbreviation for $\pair{T_1, \dotsc, T_k}$, and let $e$ be an edge pointing from some node reachable from $T_1$ via undeleteable edges to a node outside $T_1 \cup F$.  Suppose that the smallest $(T_1, T_2 \cup \dotsb \cup T_k)$-separator with undeletable edges $F$ is the same size as the smallest $(T_1, T_2 \cup \dotsb \cup T_k)$-separator with undeletable edges $F \cup \{e\}$ and has cardinality at most $p$.  Then the smallest linear cut among the terminal sets $\pair{T_1, \dotsc, T_k}$ with undeletable edges $F$ in $(V,E)$ has the same size as the smallest linear cut among these same terminals with undeletable edges $F \cup \{e\}$.
\end{thm}
\begin{proof}
First note that making edges undeleteable can only increase the size of the smallest cut.  Hence it suffices to show, under the hypothesis of the theorem, that the smallest linear cut with forbidden edges $F \cup \{e\}$ is no bigger than a minimal linear cut with forbidden edges $F$.

Let $S$ be a minimal $(T_1, T_2 \cup \dotsb \cup T_k)$-separator with undeletable  edges $F \cup \{e\}$.  Then $S$ is also a separator between these same sets with undeletable edges $F$, and by the assumption of the theorem $S$ is also a minimal such separator.  Let $W$ be a minimal linear cut in $G = (V,E)$ for $\T$ with undeletable edges $F$, and let $R$ denote the set of edges  that are reachable from $T_1$ in $(V, E \sm S)$.  We shall show that $W' = (W \cup S) \sm R$ is a linear cut in $G$ for $\T$ with undeletable edges $F \cup \{e\}$ which is no larger than $W$.  Since making edges undeletable can only increase the size of a smallest solution, $W'$ will indeed be minimal.

For clarity, we reformulate the problem instance without undeletable edges.  We replace each undeletable edge $(x,y) \in F \cup \{e\}$ with $p+1$ new, regular edges from $x$ to $y$, whereby transforming the graph into a multigraph without any undeletable edges.  Now any linear cut (resp.\ $(T_1, T_2 \cup \dotsb \cup T_k)$-separator) consisting of at most $p$ edges will be a solution in the transformed multigraph if and only if it is a solution in the original graph.  The reason is that there are not enough total cuts in the instance to sever connectivity between any vertices with $p+1$ multiedges.  Thus these edges are effectively undeletable, and of course cuts not involving undeletable edges or multiedges will work the same in both the original and transformed instance.

First we argue that $W'$ is not larger than $W$ by proving $\size{S \sm W} \leq \size{W \cap R}$.  Since $S$ does not contain any of the undeletable, multiedge parts of $G$, by Menger's Theorem \cite[Theorem~7.45]{KT06}, or more precisely its generalization to sets of vertices \cite[Lemma~1]{CLL09}, there are $\size{S}$ disjoint edge paths from $T_1$ to $\bigcup_{j > 1} T_j$, each containing an edge in $S$.  It follows that there are $\size{S \sm W}$ disjoint edge paths from $T_1$  to $S \sm W$.  Now suppose that $\size{W \cap R} < \size{S \sm W}$.  Then there must be a path from $T_1$ to some edge $x \in S \sm W$ which avoids $W \cap R$.  Furthermore, by minimality of $S$, there is a path from $x$ to some terminal set $T_j$ with $j > 1$.  But now there is a path from $T_1$ to some $T_j$ which avoids $W$, contradicting that $W$ is a linear cut.

It remains to show that $W'$ is in fact a linear cut  in $G$  for $\T$ with undeletable edges $F \cup \{e\}$.  Let $Q$ be a forbidden path.  If $Q$ does not intersect $R$, then it must pass through $W \sm R$ and hence through $W'$.  On the other hand, suppose that $Q$ does pass through $R$.  Since $T_1$ is the least-indexed terminal set, $Q$ must end at $T_j$ for some $j > 1$, and therefore $Q$ must pass through $S \subseteq W'$.  In either case, removing $W'$ eliminates the forbidden path $Q$.
\end{proof}

\begin{thm} \label{thm: ewok}
Algorithm~\ref{alg: lc} finds a solution in time $O[4^p p\cdot (\size{V} + \size{E}) \cdot \size{E}]$, if one exists, and outputs NO otherwise.
\end{thm}
\begin{proof}
Line~13 of Algorithm~\ref{alg: lc} selects an edge $e \in E \sm F$ for consideration.  If the condition for edge~$e$ in line~14 holds, then preserving $e$ does not hurt the $(T_1, T_2 \cup \dotsb \cup T_k)$-separator, and therefore by Theorem~\ref{thm: bobo} no harm comes to the \linearcut instance by adding $e$ to the list of undeletable edges.  If this condition is not satisfied, then the algorithm exhaustively searches both for a solution containing the edge~$e$ (Option~1) and for a solution not containing~$e$ (Option~2).  In Option~1, the algorithm searches for a solution of size $p-1$ containing $e$, and in Option~2, the size of the smallest$(T_1, T_2 \cup \dotsb \cup T_k)$-separator increases by~1.  Along any branch of the algorithm, either of these two Options can occur at most $p$ times for each terminal before the algorithm returns \NO, and the latter happens only when exhaustive search fails to find a solution.  Hence the algorithm eventually terminates with the correct answer.

We can refine our analysis further to show that there are at most $4^p$ possible branches in the algorithm.  We argue that any branch of the algorithm witnesses at most $2p$ branching splits.  Suppose that the initial input parameter is $p$ and that the smallest $(T_1, T_2 \cup \dotsb \cup T_k)$-separator has size $m$.  Since each iteration of Option~1 decreases the size of the minimal $(T_1, T_2 \cup \dotsb \cup T_k)$-separator by~1, the path which always chooses Option~1 will witness exactly~$m$ branches up to the point where Line~11 of Algorithm~\ref{alg: lc} recognizes that $T_1$ has been separated and removes it from further consideration.  Each time Option~2 is chosen along the path, the size of the smallest$(T_1, T_2 \cup \dotsb \cup T_k)$-separator increases by at least~1, so if Option~2 happens $r$ times, then Option~1 must happen a total of at least $m + r$ times before $T_1$ is separated.  Thus the size of the cutset size parameter when $T_1$ becomes separated is at most $p - m - r$, the initial parameter value minus the number of times Option~1 was chosen, and the total number of splits witnessed is $(m+r) + r$, which is at most twice the number of edges added to the cutset.  The same counting argument holds for separators for successive $T_i$'s and it follows that each search path can witness at most $2p$ splits in case the algorithm succeeds.

The number of steps between each encounter with an Option is essentially the time required to check whether a separator size~$p$ exists, which is $O[p (\size{V} + \size{E})]$ by the argument in \cite[Lemma~2]{CLL09}, times the number edges. The multiplicative factor of $\size{E}$ comes from the potential recursion in line~15.  Hence the total runtime is $O[2^{2p} p \cdot (\size{V} + \size{E}) \cdot \size{E}]$.
\end{proof}

\begin{cor}
\linearcut is fixed-parameter tractable when parameterized in the size of the cutset.
\end{cor}

\section{Hardness result}

Marx and Razgon \cite{MR11} showed that \dmulticut parameterized in the size of the cutset is $\W[1]$-hard by reducing this problem to the known $\W[1]$-hard problem $\clique$.  Therefore the following is immediate from Theorem~\ref{thm: hop}.
\begin{cor}
\latticecut is $\W[1]$-hard when parameterized in the size of the cutset.
\end{cor}

Whether $\dmulticut$ is fixed-parameter tractable when parameterized in both the size of the cutset and the number of terminals remains an open problem, even in the case where we fix the number of terminal pairs at $k=3$ \cite{CHM12, MR11}.  $\linearcut$ for $k=2$ is possible via the Ford-Fulkerson algorithm, however for longer chains the problem also becomes $\NP$-hard.
\begin{prop} \label{prop: lcnp}
Deciding whether a \linearcut instance has a solution of size $p$ is $\NP$-complete for $k=3$ terminals.
\end{prop}

\begin{proof}
\linearcut is trivially in $\NP$ as one can easily check by breadth-first search whether a given set of edges is a solution.

We reduce the undirected \umultiwaycut problem for $k=3$, which is $\NP$-hard \cite{DJPSY92}, to the \linearcut problem for $k=3$.  Let $G$ be an undirected graph with terminal nodes $s$, $t$ and $u$ be an instance of \umultiwaycut, the problem of finding a smallest edge cut which separates $s$, $t$, and $u$.  Construct a new directed graph $G'$ which has the same vertices as $G$ except for each edge $e=\{x,y\}$ in $G$ we also add two new vertices $a_e$ and $b_e$.  The edges from $G$ do not carry over to $G'$, and instead we add directed edges $(x, a_e)$, $(y, a_e)$, $(a_e, b_e)$, $(b_e, y)$, and $(b_e , x)$.  We call this collection of edges the \emph{gadget} for $e$.  Our \linearcut instance consists of the graph $G'$ together with the embedded terminals nodes  $s$, $t$, and $u$ from $G$ with the (arbitrary) tuple ordering $\pair{s,t,u}$.  Technically we treat the terminal nodes here as singleton sets when formulating this instance of \linearcut.

Assume $C = \{e_1, \dotsc, e_p\}$ is a \umultiwaycut solution for $G$.  We claim that $C' = \{(a_{e_1}, b_{e_1}), \dotsc, (a_{e_p}, b_{e_p})\}$ is then a \linearcut solution for $G'$.  Suppose there were some prohibited path in $G'$ between two terminals, say $s$ and $t$, which avoids $C'$.  This path must have the form
\[
s \to a_{(s,x_1)} \to b_{(s,x_1)} \to x_1 \to a_{(x_1,x_2)} \to b_{(x_1,x_2)} \to x_2 \to \dotsb \to t
\]
for some vertices $x_1, x_2, \dotsc$ in $G$.  Contracting all the $a_i$'s and $b_i$'s from this path yields a path from $s$ to $t$ in $G$ which avoids $C$, which is impossible.

Conversely, assume that $C' = \{d_1, \dotsc, d_p\}$ is a \linearcut solution for $G$.  For each $i \leq p$, let $e_i$ be the gadget for the edge in $G$ which $d_i$ belongs to.  Then $C = \{e_1, \dotsc, e_p\}$ is a \umultiwaycut solution for $G$ as any path $x_1 \to  \dotsb \to  x_k$ between terminals in $G$ avoiding $C$ gives rise to a path between the same terminals in $G'$ which avoids $C'$, namely
\[
x_1 \to a_{(x_1,x_2)} \to b_{(x_1,x_2)} \to x_2 \to a_{(x_2, x_3)} \to b_{(x_2,x_3)} \to x_3 \to \dotsb \to x_k,
\]
which cannot exist.  Thus \umultiwaycut is polynomial-time reducible to \linearcut.
\end{proof}

\section{Approximation}

It seems difficult to efficiently approximate \dmulticut \cite{AACM07, CK06, Gup03}, which indicates that \latticecut may not have a good approximation algorithm either.  The best known polynomial-time approximation algorithm for \dmulticut is just under $O(\sqrt n)$ \cite{AACM07}.  We wonder whether \linearcut may be easier to approximate.

Recall that $\dmultiwaycut$ is the problem of \latticecut restricted to the instances  where the underlying poset is an antichain.  

\begin{framed}
\begin{description}
\item[Name:] \dmultiwaycut

\item[Instance:] A directed graph $(V,E)$ and a tuple of \emph{terminal} sets $T_1, \dotsc, T_k$ which are subsets of $V$.

\item[Problem:] Find a smallest set of edges $S \subseteq E$ such that there is no path from $T_i$ to $T_j$ in $(V, E\sm S)$ for all $i \neq j$.
\end{description}
\end{framed}

Garg, Vazirani, and Yannakakis \cite{GVY94} gave a $2 \log n$ approximation for \dmultiwaycut, later improved to a factor of 2 by Naor and Zosin \cite{NZ97} using an LP relaxation.  The undirected \umultiwaycut problem for $k$~terminals has a simple $2 - 2/k$ approximation algorithm using isolated cuts \cite{DJPSY92} and even a $1.5 - 2/k$ approximation using LP relaxation \cite{CKR98} (see also \cite{Vaz03}).  By making two calls to Algorithm~\ref{alg: lc}, we can obtain a simple approximation to \dmultiwaycut which runs faster than Chitnis, Hajiaghayi, and Marx's $2^{2^{O(p)}}n^{O(1)}$-time exact solution \cite{CHM12} but does not beat Naor and Zosin's polynomial-time 2-approximation \cite{NZ97}.

\begin{cor} \label{cor: 2chm}
One can find a solution for \dmultiwaycut of instance size~$n$ in time $O(4^p p n^4)$ which is within a factor of two of optimal whenever a solution of size~p exists.
\end{cor}
\begin{proof}
Assume that $T_1, \dotsc, T_k$ are the terminal sets which need to be separated in the directed  graph $(V,E)$.  Using Algorithm~\ref{alg: lc}, make one \linearcut which cuts using the terminal sets $\pair{T_1,\dotsc, T_k}$ and another which uses this $k$-tuple reversed, $\pair{T_k, \dotsc, T_1}$.  The union of these two cuts is a solution to the \dmultiwaycut instance, when both exist, and neither cut is larger than the smallest possible solution.
\end{proof}

\bibliographystyle{plain}
\bibliography{lattice_cut}

\end{document}